\documentclass[11pt]{article}

\usepackage[margin=1in]{geometry}
\usepackage{amsmath,amssymb,amsthm,mathtools}
\usepackage{microtype}
\usepackage[hidelinks]{hyperref}

\newtheorem{theorem}{Theorem}
\newtheorem{lemma}[theorem]{Lemma}
\newtheorem{proposition}[theorem]{Proposition}
\newtheorem{example}[theorem]{Example}
\newtheorem{definition}[theorem]{Definition}

\newcommand{\Z}{\mathbb Z}
\newcommand{\PN}{\mathcal P_{n,N}}
\newcommand{\pos}{\operatorname{pos}}

\title{An Elementary Proof of the \(\widetilde O(n^{1/3})\) Bound for Separating Words}
\author{Chen Xu}
\date{}

\begin{document}

\maketitle

\begin{abstract}
For two distinct binary words of length \(n\), the separating words problem asks
for a small deterministic finite automaton that accepts exactly one of them.
Chase proved a \(\widetilde O(n^{1/3})\) upper bound using a complex-analytic
estimate for sparse polynomials. We replace that estimate by a finite-difference
argument and a second-order real recurrence cutoff. The resulting elementary
proof gives an explicit bound of \(O(n^{1/3}(\log n)^{7/3})\) states.
\end{abstract}

\section{Introduction}

For distinct binary words \(x,y\in\{0,1\}^n\), let
\(\operatorname{sep}(x,y)\) be the minimum number of states in a deterministic
finite automaton that accepts exactly one of \(x\) and \(y\). Define
\[
  S(n)=\max_{x\ne y\in\{0,1\}^n}\operatorname{sep}(x,y).
\]
The separating words problem asks for asymptotic upper and lower bounds on
\(S(n)\).

\paragraph{Previous work.}
The problem was introduced by Goralčík and Koubek
\cite{goralcik-koubek}; see also the survey of Demaine, Eisenstat, Shallit, and
Wilson \cite{demaine-eisenstat-shallit-wilson}. Robson proved
\(S(n)=O(n^{2/5}(\log n)^{3/5})\) \cite{robson1989} and later studied related
machine and group formulations \cite{robson1996}. Chase obtained
\(S(n)=O(n^{1/3}\log^7 n)\) by reducing the problem to separated sets of
occurrence positions \cite{chase}. Lower bounds are connected with short identities in
finite transformation semigroups \cite{bulatov-karpova-shur-startsev}.

\paragraph{Contribution.}
Our contribution is an alternative direct proof of the sparse-polynomial
divisibility estimate underlying Chase's \(\widetilde O(n^{1/3})\) bound
\cite{chase}. Chase obtains this estimate from a complex-analytic statement;
we instead prove it directly using a finite-difference characterization of
divisibility by \((1-x)^K\) and a second-order real recurrence that separates
an isolated constant term from a distant coefficient tail. Once this estimate
is established, the reduction to a separating DFA follows Chase's argument and
is included here for completeness. Tracking the parameters gives the following
explicit bound.

\begin{theorem}\label{thm:main}
\[
  S(n)=O\!\left(n^{1/3}(\log n)^{7/3}\right).
\]
\end{theorem}

The direct proof of the divisibility estimate is the new ingredient; the
subsequent DFA construction is due to Chase. We do not seek a better power of
\(n\): the separated-set framework has an \(n^{1/3}\)-scale barrier
\cite{chase}.

\paragraph{Proof outline.}
We first find a block whose occurrence positions in the two words form distinct
\(N\)-separated sets. Their difference is encoded by a sparse polynomial.
The central divisibility estimate shows that this polynomial cannot contain a
large factor \((1-x)^K\). A low-degree power-sum difference then yields a
residue class modulo a prime \(p=O(N)\), and a modular occurrence-counting DFA
separates the words. Section~\ref{sec:prop3} proves the divisibility estimate by
finite differences and a real recurrence cutoff. The final section compares the
resulting parameter choices with the previous proof.

\section{Main Proof}

Throughout the proof, all logarithms are natural. We assume \(n\) is
sufficiently large; the remaining values are absorbed into the implicit
constant in Theorem~\ref{thm:main}. Set
\[
  N=\left\lceil n^{1/3}(\log n)^{4/3}\right\rceil .
\]
For large \(n\), we have \(1\le N\le n/3\). We write
\([n]=\{0,1,\ldots,n-1\}\). For a word \(x\) and a block \(w\), let
\[
  \pos_w(x)=\{j:x_jx_{j+1}\cdots x_{j+|w|-1}=w\}.
\]
A set \(A\subseteq[n]\) is \(N\)-separated if any two distinct elements of
\(A\) differ by at least \(N\). For a modulus \(p\) and residue \(i\), write
\[
  A_{i,p}=\{a\in A:a\equiv i\pmod p\}.
\]

\subsection{From words to separated occurrences}

We first record the elementary periodicity fact used to create separated
occurrence sets.

\begin{definition}[Period]
A positive integer \(p\le L\) is a \emph{period} of a word
\(z=z_0z_1\cdots z_{L-1}\) if \(z_i=z_{i+p}\) for every \(0\le i<L-p\).
\end{definition}

\begin{lemma}\label{lem:overlap}
If \(p\) and \(q\) are periods of a word \(z\) of length \(L\) and
\(L\ge p+q-\gcd(p,q)\), then \(\gcd(p,q)\) is also a period.
\end{lemma}

\begin{proof}
Suppose \(\gcd(p,q)\) is not a period, so
\(z_i\ne z_{i+\gcd(p,q)}\) for some \(i\). Choose \(k,\ell\ge0\) with
\(kp-\ell q=\gcd(p,q)\). Since \(p\) and \(q\) are periods,
\(z_i=z_{i+kp}=z_{i+\ell q+\gcd(p,q)}=z_{i+\gcd(p,q)}\), a contradiction.
\end{proof}

\begin{lemma}\label{lem:twoextension}
For \(v\in\{0,1\}^{2N-1}\), either \(v0\) or \(v1\) has no period \(<N\).
\end{lemma}

\begin{proof}
Suppose \(v0\) and \(v1\) have periods \(p,q<N\). The common prefix \(v\) has
periods \(p,q\), and \(p+q-\gcd(p,q)<2N-1=|v|\), so
Lemma~\ref{lem:overlap} implies that \(\gcd(p,q)\) is a period of \(v\). Since
\(2N-1-p\equiv2N-1-q\pmod{\gcd(p,q)}\),
\(0=(v0)_{2N-1}=v_{2N-1-p}=v_{2N-1-q}=(v1)_{2N-1}=1\), a contradiction.
\end{proof}

\begin{lemma}\label{lem:separated}
For distinct \(x,y\in\{0,1\}^n\), either \(x_i\ne y_i\) for some \(i<2N\), or
there is a \(w\in\{0,1\}^{2N}\) such that \(\pos_w(x)\ne\pos_w(y)\) and both
sets are \(N\)-separated.
\end{lemma}

\begin{proof}
Let \(k=\min\{i:x_i\ne y_i\}\). If \(k<2N\), the first alternative holds.
Otherwise let
\(v=x_{k-2N+1}\cdots x_{k-1}=y_{k-2N+1}\cdots y_{k-1}\). By
Lemma~\ref{lem:twoextension}, choose \(w\in\{v0,v1\}\) with no period \(<N\).
Since \(x_k\ne y_k\), \(w\) occurs at position \(k-2N+1\) in exactly one of
\(x,y\), so \(\pos_w(x)\ne\pos_w(y)\). If two occurrences of \(w\) in either
word began \(d<N\) positions apart, then \(d\) would be a period of \(w\), a
contradiction. Hence both occurrence sets are \(N\)-separated.
\end{proof}

\subsection{The sparse-polynomial reduction}

Let \(A,B\subseteq[n]\) be distinct \(N\)-separated sets. Define
\[
  f_{A,B}(x)=\sum_{a\in A}x^a-\sum_{b\in B}x^b.
\]
Let \(r\) be the least exponent whose coefficient in \(f_{A,B}\) is nonzero.
After changing the sign of \(f_{A,B}\) if necessary, assume this coefficient is
\(+1\), and set
\[
  F(x)=x^{-r}f_{A,B}(x).
\]
Then \(F(0)=1\). The first nonzero position \(r\) lies in \(A\setminus B\).
Degrees \(1,\ldots,N-1\) of \(F\) correspond to positions
\(r+1,\ldots,r+N-1\). Since \(A\) is \(N\)-separated, there is no other
positive term in this range. Since \(B\) is \(N\)-separated, there is at most
one negative term in this range. Therefore the low-degree part of \(F\) is
either empty or exactly \(-x^d\) for some \(1\le d<N\).

This motivates the following two families:
\[
  \PN^{(0)}
  =
  \left\{1+\sum_{j=N}^n a_jx^j:\ |a_j|\le1\right\},
\]
and
\[
  \PN^{(1)}
  =
  \left\{1-x^d+\sum_{j=N}^n a_jx^j:\ 1\le d<N,\ |a_j|\le1\right\}.
\]
We write
\[
  \PN=\PN^{(0)}\cup\PN^{(1)}.
\]
The normalized polynomial \(F\) always lies in \(\PN\).

\begin{example}
For \(N=5\), let \(A=\{3,11,19\}\) and \(B=\{7,15\}\). Then
\[
  f_{A,B}(x)=x^3-x^7+x^{11}-x^{15}+x^{19},
\]
and after shifting by \(x^{-3}\),
\[
  F(x)=1-x^4+x^8-x^{12}+x^{16}\in\PN^{(1)}.
\]
\end{example}

The technical input is the following divisibility estimate.

\begin{proposition}\label{prop:div}
No \(P\in\PN\) has a factor \((1-x)^K\) with
\(K\ge8\sqrt{n/N}\log n\).
\end{proposition}

We prove Proposition~\ref{prop:div} in Section~\ref{sec:prop3}. First we show
how it finishes the main theorem.

\subsection{Divisibility as finite differences}

For \(P(x)=\sum_jc_jx^j\) and \(R(t)\in\mathbb R[t]\), define
\[
  \langle P,R\rangle=\sum_jc_jR(j).
\]
For example, \(R(t)=t^2\) gives \(\langle P,R\rangle=\sum_jc_j\,j^2\).

\begin{lemma}\label{lem:divtest}
A polynomial \(P\) is divisible by \((1-x)^K\) if and only if
\(\langle P,R\rangle=0\) for every polynomial \(R\) of degree \(<K\).
\end{lemma}

\begin{proof}
Define the finite difference
\[
  \Delta R(t)=R(t)-R(t+1).
\]
By induction on \(K\), this gives
\begin{equation}\label{eq:finite-difference-expansion}
  \Delta^K R(t)
  =
  \sum_{r=0}^K(-1)^r\binom Kr R(t+r).
\end{equation}
Indeed, the case \(K=1\) is the definition. If the formula holds for \(K\),
then
\[
\begin{aligned}
  \Delta^{K+1}R(t)
  &=\Delta^KR(t)-\Delta^KR(t+1)\\
  &=
  \sum_{r=0}^K(-1)^r\binom Kr R(t+r)
  -
  \sum_{r=0}^K(-1)^r\binom Kr R(t+1+r).
\end{aligned}
\]
After shifting the index in the second sum, the coefficient of \(R(t+r)\) is
\[
  (-1)^r\binom Kr-(-1)^{r-1}\binom K{r-1}
  =
  (-1)^r\binom{K+1}r,
\]
with the usual convention that \(\binom K{-1}=\binom K{K+1}=0\). This proves
\eqref{eq:finite-difference-expansion}.

Now assume \(P=(1-x)^KQ\), and write
\[
  Q(x)=\sum_\ell q_\ell x^\ell .
\]
Then
\[
\begin{aligned}
  P(x)
  &=
  \sum_\ell q_\ell x^\ell(1-x)^K  \\
  &=
  \sum_\ell q_\ell x^\ell
  \sum_{r=0}^K(-1)^r\binom Kr x^r \\
  &=
  \sum_\ell\sum_{r=0}^K
  q_\ell(-1)^r\binom Kr x^{\ell+r}.
\end{aligned}
\]
Since \(\langle x^s,R\rangle=R(s)\), linearity gives
\[
\begin{aligned}
  \langle P,R\rangle
  &=
  \sum_\ell\sum_{r=0}^K
  q_\ell(-1)^r\binom Kr R(\ell+r)\\
  &=
  \sum_\ell q_\ell
  \left(\sum_{r=0}^K(-1)^r\binom Kr R(\ell+r)\right)\\
  &=
  \sum_\ell q_\ell\,\Delta^KR(\ell),
\end{aligned}
\]
where the last equality is \eqref{eq:finite-difference-expansion}. The operator
\(\Delta\) lowers degree by one, so \(\Delta^KR=0\) whenever \(\deg R<K\).
This proves the first implication.

For the converse, assume all pairings against degree \(<K\) polynomials vanish.
Write \(P(x)=\sum_jc_jx^j\). For \(0\le r<K\), take
\[
  R_r(t)=t(t-1)\cdots(t-r+1),
\]
with \(R_0(t)=1\). Then
\[
  \langle P,R_r\rangle
  =
  \sum_j c_jj(j-1)\cdots(j-r+1)
  =
  P^{(r)}(1).
\]
Thus
\[
  P(1)=P'(1)=\cdots=P^{(K-1)}(1)=0.
\]
This is exactly the statement that \(P\) has a factor \((x-1)^K\), equivalently
a factor \((1-x)^K\).
\end{proof}

\subsection{Completing the automaton construction}

Let \(K_0=\left\lceil8\sqrt{n/N}\log n\right\rceil\). Since
\(f_{A,B}=\pm x^rF\), Proposition~\ref{prop:div} applied to \(F\) gives
\((1-x)^{K_0}\nmid f_{A,B}\). By Lemma~\ref{lem:divtest}, there is a
polynomial \(R\) of degree \(<K_0\) with
\[
  0\ne \langle f_{A,B},R\rangle
  =
  \sum_{a\in A}R(a)-\sum_{b\in B}R(b).
\]
Write \(R\) in the monomial basis:
\[
  R(t)=\sum_{m=0}^{K_0-1}\alpha_m t^m.
\]
If every power-sum difference
\[
  \sum_{a\in A}a^m-\sum_{b\in B}b^m
\]
were zero for \(0\le m<K_0\), then \(\langle f_{A,B},R\rangle\) would be zero.
Therefore there is an exponent
\[
  m<K_0=O(\sqrt{n/N}\log n)
\]
such that
\begin{equation}\label{eq:momentdiff}
  \sum_{a\in A}a^m\ne\sum_{b\in B}b^m.
\end{equation}

Let
\[
  D=\sum_{a\in A}a^m-\sum_{b\in B}b^m.
\]
Then \(D\ne0\). At most \(n\) positions can contribute to this difference, and
each contributing term has absolute value at most \(n^m\). Hence
\[
  |D|\le n\cdot n^m=n^{m+1}
  =\exp((m+1)\log n).
\]
Set
\[
  E=|D|\prod_{\substack{\ell\le2N\\\ell\text{ prime}}}\ell.
\]
By the prime number theorem, there is a prime \(p=O(\log E)\) such that
\(p\nmid E\). Moreover,
\[
  \log E=\log|D|+\sum_{\substack{\ell\le2N\\\ell\text{ prime}}}\log \ell
  =O((m+1)\log n+N)=O(N),
\]
where the last equality follows from \(m<K_0\) and the choice of \(N\).
Thus \(p=O(N)\), while \(p\nmid E\) gives \(p>2N\) and \(p\nmid D\).

Reducing \eqref{eq:momentdiff} modulo \(p\), we get
\[
  0\ne D
  \equiv
  \sum_{i=0}^{p-1}
  \bigl(|A_{i,p}|-|B_{i,p}|\bigr)i^m
  \pmod p.
\]
Therefore some residue \(i\in\Z/p\Z\) satisfies
\[
  |A_{i,p}|\ne |B_{i,p}|.
\]

Now return to the original words \(x,y\). If they differ among their first
\(2N\) symbols, then a prefix-checking DFA with \(O(N)\) states separates them.
Otherwise Lemma~\ref{lem:separated} gives a block \(w\in\{0,1\}^{2N}\) such
that
\[
  A=\pos_w(x),\qquad B=\pos_w(y)
\]
are distinct \(N\)-separated sets. The preceding paragraph gives a prime
\(p=O(N)\), with \(p>|w|=2N\), and a residue \(i\) such that
\[
  |\pos_w(x)_{i,p}|\ne|\pos_w(y)_{i,p}|.
\]
The two counts are between \(0\) and \(n\), so their nonzero difference has
absolute value at most \(n\). If its absolute value is \(1\), take \(q=2\);
otherwise the same consequence of the prime number theorem gives a prime
\(q=O(\log n)\) that does not divide it. Thus the two counts differ modulo
\(q\).

We spell out the occurrence-counting DFA.

\begin{lemma}\label{lem:countdfa}
Let \(w\) be a block with \(|w|<p\), let \(i\in\Z/p\Z\), and let \(q\) be a
positive integer. There is a DFA with \(2pq\) states whose final counter equals
\[
  |\{j\in\pos_w(x):j\equiv i\pmod p\}|\pmod q
\]
after reading a word \(x\).
\end{lemma}

\begin{proof}
Use states
\[
  \Z/p\Z\times\{0,1\}\times\Z/q\Z.
\]
The first coordinate is the current position modulo \(p\). The second
coordinate is an ``alive'' bit. It records whether the candidate occurrence of
\(w\) that began at the most recent position congruent to \(i\) is still
consistent with the letters read so far. The third coordinate is the number of
completed occurrences modulo \(q\).

Because \(|w|<p\), there can be at most one candidate occurrence starting in
the residue class \(i\pmod p\) alive at any moment. If the current clock value
is \(r\), then the offset inside the candidate, when a candidate is alive, is
the unique integer \(\ell\in\{0,\ldots,|w|-1\}\) with
\[
  r\equiv i+\ell\pmod p.
\]
Thus the clock determines which letter \(w_\ell\) should be tested. If the input
letter disagrees with \(w_\ell\), the alive bit is turned off. If
\(\ell=|w|-1\) and the letter agrees, the occurrence is completed and the third
coordinate is incremented modulo \(q\). When the clock is \(i\), the automaton
starts a new candidate if the current input letter agrees with \(w_0\).
This implements exactly the stated residue-class occurrence count.
\end{proof}

By Lemma~\ref{lem:countdfa}, the required DFA has states
\[
  \Z/p\Z\times\{0,1\}\times\Z/q\Z.
\]
At the end, choose the accepting counter residues so that the DFA accepts one of
\(x,y\) and rejects the other. The number of states is
\[
  2pq=O(N\log n)
  =
  O\!\left(n^{1/3}(\log n)^{7/3}\right),
\]
which proves Theorem~\ref{thm:main}.

\section{Proof of Proposition~\ref{prop:div}}\label{sec:prop3}

We restate the proposition.

\noindent\textbf{Proposition~\ref{prop:div}.}
No \(P\in\PN\) has a factor \((1-x)^K\) with
\(K\ge8\sqrt{n/N}\log n\).

\subsection{A bounded recurrence cutoff}

The next step constructs low-degree test polynomials that contradict the
divisibility test.

\begin{lemma}\label{lem:cutoff}
Define polynomials \(\Phi_\ell(z)\) by
\[
  \Phi_0(z)=1,\qquad \Phi_1(z)=z,\qquad
  \Phi_{\ell+1}(z)=2z\Phi_\ell(z)-\Phi_{\ell-1}(z).
\]
Then:
\[
  |\Phi_\ell(z)|\le1\qquad(-1\le z\le1),
\]
and, for \(0<\eta\le1\),
\[
  \Phi_\ell(1+\eta)\ge \frac12 e^{\ell\sqrt{\eta}/2}.
\]
Moreover \(\Phi_\ell\) has parity \(\ell\), meaning
\(\Phi_\ell(-z)=(-1)^\ell\Phi_\ell(z)\).
\end{lemma}

\begin{proof}
The parity statement follows immediately from the recurrence and the initial
conditions.

We first prove the boundedness on \([-1,1]\) by a real invariant. Fix
\(z\in[-1,1]\). We claim that
\[
  \Phi_k(z)^2-2z\Phi_k(z)\Phi_{k-1}(z)
  +\Phi_{k-1}(z)^2=1-z^2
\]
for every \(k\ge1\). For \(k=1\), this is
\[
  z^2-2z^2+1=1-z^2.
\]
If the identity holds at \(k\), then using
\[
  \Phi_{k+1}(z)=2z\Phi_k(z)-\Phi_{k-1}(z),
\]
we get
\[
\begin{aligned}
  &\Phi_{k+1}(z)^2-2z\Phi_{k+1}(z)\Phi_k(z)+\Phi_k(z)^2 \\
  &\qquad=(2z\Phi_k(z)-\Phi_{k-1}(z))^2
    -2z(2z\Phi_k(z)-\Phi_{k-1}(z))\Phi_k(z)+\Phi_k(z)^2  \\
  &\qquad=\Phi_k(z)^2-2z\Phi_k(z)\Phi_{k-1}(z)
    +\Phi_{k-1}(z)^2.
\end{aligned}
\]
So the invariant holds for all \(k\). If \(|z|<1\), rewrite it as
\[
  (\Phi_k(z)-z\Phi_{k-1}(z))^2
  +(1-z^2)\Phi_{k-1}(z)^2=1-z^2.
\]
Hence \(|\Phi_{k-1}(z)|\le1\) for every \(k\), and therefore
\(|\Phi_k(z)|\le1\). At the endpoints \(z=1\) and \(z=-1\), the recurrence
gives \(\Phi_k(1)=1\) and \(\Phi_k(-1)=(-1)^k\), so the same bound holds.

For the growth estimate outside \([-1,1]\), use the real characteristic roots.
Fix \(z>1\), and put
\[
  \rho=z+\sqrt{z^2-1}.
\]
The second root of
\[
  \lambda^2-2z\lambda+1=0
\]
is \(z-\sqrt{z^2-1}\). Since the product of the two roots is \(1\), this second
root is \(\rho^{-1}\), and hence
\[
  \rho+\rho^{-1}=2z.
\]
The sequence
\[
  \frac{\rho^\ell+\rho^{-\ell}}2
\]
has initial values (1,z) and satisfies the same recurrence as
\(\Phi_\ell(z)\). Hence
\[
  \Phi_\ell(z)=\frac{\rho^\ell+\rho^{-\ell}}2.
\]

Now take \(z=1+\eta\), where \(0<\eta\le1\). Then
\[
  \rho
  =
  1+\eta+\sqrt{(1+\eta)^2-1}
  =
  1+\eta+\sqrt{2\eta+\eta^2}
  \ge 1+\sqrt{\eta}.
\]
Therefore
\[
  \Phi_\ell(1+\eta)
  =
  \frac{\rho^\ell+\rho^{-\ell}}2
  \ge
  \frac12(1+\sqrt{\eta})^\ell.
\]
Since \(\log(1+s)\ge s/2\) for \(0\le s\le1\), this gives
\[
  \Phi_\ell(1+\eta)\ge \frac12 e^{\ell\sqrt{\eta}/2}.
\]
These \(\Phi_\ell\)'s are the usual Chebyshev polynomials of the first kind
\cite{rivlin}, but the proof above used only the recurrence, the real invariant,
and the real root formula.
\end{proof}

We now return to the sparse polynomials in Proposition~\ref{prop:div}. Their
tail terms occur only in degrees \(N,N+1,\ldots,n\). The recurrence cutoff will
be applied after mapping this tail interval to \([-1,1]\), where
Lemma~\ref{lem:cutoff} keeps \(\Phi_\ell\) bounded. The isolated constant term
at degree \(0\) is mapped just outside \([-1,1]\), where the same recurrence
polynomial grows. This is the mechanism that separates the constant term from
the tail.

Define the affine map
\[
  u(t)=\frac{2t-(n+N)}{n-N}.
\]
Then
\[
  u(N)=-1,\qquad u(n)=1,
\]
so
\[
  u([N,n])=[-1,1].
\]
At the isolated point \(0\),
\[
  u(0)=-\frac{n+N}{n-N}
      =-1-\frac{2N}{n-N}.
\]
Thus
\[
  |u(0)|=1+\eta,
  \qquad
  \eta=\frac{2N}{n-N}.
\]
Since \(N\le n/3\), we have
\[
  \frac{2N}{n}\le \eta\le \frac{3N}{n}\le1,
\]
so \(\eta\asymp N/n\).

Define the cutoff degree
\[
  \kappa(n,N)=\left\lceil6\sqrt{n/N}\log n\right\rceil .
\]
By Lemma~\ref{lem:cutoff}, the parity of \(\Phi_\ell\), and
\(\eta\ge2N/n\),
\begin{equation}\label{eq:denomlarge}
\begin{aligned}
  |\Phi_{\kappa(n,N)}(u(0))|
  &=\Phi_{\kappa(n,N)}(1+\eta)\\
  &\ge \frac12e^{\kappa(n,N)\sqrt{\eta}/2}\\
  &\ge \frac12e^{3\sqrt2\log n}
   =\frac12n^{3\sqrt2}\ge n^4
\end{aligned}
\end{equation}
for sufficiently large \(n\).
This calculation also explains the scale of the cutoff. Since the distance of
\(u(0)\) from \([-1,1]\) is \(\eta\asymp N/n\), a degree-\(\ell\) recurrence
polynomial grows there like
\[
  \exp(\Theta(\ell\sqrt{N/n})).
\]
To obtain polynomial size in \(n\), one naturally needs
\[
  \ell=\Theta(\sqrt{n/N}\log n).
\]

\subsection{Separating the constant term from the tail}

In the next two lemmas, the subscript \(\kappa(n,N)\) emphasizes that the
cutoff degree depends on the two parameters \(n\) and \(N\).

We first treat \(\PN^{(0)}\).

\begin{lemma}\label{lem:p0}
No polynomial
\[
  P(x)=1+\sum_{j=N}^n a_jx^j,\qquad |a_j|\le1,
\]
is divisible by \((1-x)^K\) for \(K>\kappa(n,N)\).
\end{lemma}

\begin{proof}
Use the polynomial
\[
  R(t)=\frac{\Phi_{\kappa(n,N)}(u(t))}{\Phi_{\kappa(n,N)}(u(0))}.
\]
Since \(u(t)\) is affine and \(\Phi_{\kappa(n,N)}\) has degree \(\kappa(n,N)\),
\(\deg R=\kappa(n,N)\).
Also \(R(0)=1\). For \(N\le j\le n\), we have \(u(j)\in[-1,1]\), and therefore
by Lemma~\ref{lem:cutoff},
\[
  |\Phi_{\kappa(n,N)}(u(j))|\le1.
\]
Using \eqref{eq:denomlarge},
\[
  |R(j)|\le n^{-4}
  \qquad(N\le j\le n).
\]

Suppose \((1-x)^K\mid P\) with \(K>\kappa(n,N)=\deg R\).
Lemma~\ref{lem:divtest} gives
\[
  0=\langle P,R\rangle
   =R(0)+\sum_{j=N}^n a_jR(j).
\]
The tail is bounded by
\[
  \left|\sum_{j=N}^n a_jR(j)\right|
  \le
  \sum_{j=N}^n |R(j)|
  \le
  n\cdot n^{-4}
  =
  n^{-3}
  <1.
\]
But \(R(0)=1\), so the displayed equality is impossible. Hence no such
divisibility can occur for \(K>\deg R\).
\end{proof}

It remains to handle the extra low-degree term in \(\PN^{(1)}\).

\begin{lemma}\label{lem:p1}
No polynomial
\[
  P(x)=1-x^d+\sum_{j=N}^n a_jx^j,
  \qquad 1\le d<N,\quad |a_j|\le1,
\]
is divisible by \((1-x)^K\) for \(K>\kappa(n,N)+1\).
\end{lemma}

\begin{proof}
The term \(-x^d\) is below the tail interval, so we multiply by a linear factor
that kills it. Define
\[
  R(t)=\left(1-\frac td\right)\frac{\Phi_{\kappa(n,N)}(u(t))}{\Phi_{\kappa(n,N)}(u(0))}.
\]
Then
\[
  R(0)=1,\qquad R(d)=0,
\]
and
\(\deg R\le \kappa(n,N)+1\).
For \(N\le j\le n\), using \(|\Phi_{\kappa(n,N)}(u(j))|\le1\) and
\eqref{eq:denomlarge},
\[
  |R(j)|
  \le
  \left(1+\frac jd\right)n^{-4}.
\]
Since \(d\ge1\) and \(j\le n\),
\[
  1+\frac jd\le 1+n\le2n
\]
for \(n\ge1\). Hence
\[
  |R(j)|\le2n^{-3}
  \qquad(N\le j\le n).
\]

Suppose \((1-x)^K\mid P\) with \(K>\kappa(n,N)+1\ge\deg R\).
Lemma~\ref{lem:divtest} gives
\[
  0=\langle P,R\rangle
   =R(0)-R(d)+\sum_{j=N}^n a_jR(j)
   =
   1+\sum_{j=N}^n a_jR(j).
\]
The tail is bounded by
\[
  \left|\sum_{j=N}^n a_jR(j)\right|
  \le
  n\cdot 2n^{-3}
  =
  2n^{-2}
  <1
\]
for large \(n\). This contradicts the displayed identity. Thus the desired
nondivisibility holds in the \(\PN^{(1)}\) case as well.
\end{proof}

For large \(n\), \(\sqrt{n/N}\log n>1\), and hence
\[
  \kappa(n,N)+1
  \le6\sqrt{n/N}\log n+2
  <8\sqrt{n/N}\log n.
\]
Proposition~\ref{prop:div} now follows from
Lemmas~\ref{lem:p0} and~\ref{lem:p1}.

\section{Comparison with the Previous Proof}

The previous proof uses a moment exponent of order
\(n^{1/3}\log^5 n\), a prime modulus of order \(n^{1/3}\log^6 n\), and one
additional logarithm for the final modular counter, giving an
\(O(n^{1/3}\log^7 n)\) state bound.

Our recurrence cutoff replaces the complex-analytic sparse-polynomial estimate
by
\[
  K=O\!\left(\sqrt{n/N}\log n\right).
\]
The subsequent power-sum and prime-selection steps give
\[
  p=O\!\left(N+\sqrt{n/N}\log^2 n\right),
  \qquad q=O(\log n).
\]
Hence the final DFA has
\[
  O\!\left(\left(N+\sqrt{n/N}\log^2 n\right)\log n\right)
\]
states. Balancing the two terms inside the parentheses yields
\[
  N^3=n\log^4 n,
  \qquad
  N=n^{1/3}(\log n)^{4/3},
\]
and therefore the bound in Theorem~\ref{thm:main}.

For the recurrence-cutoff method itself, the degree
\(\Theta(\sqrt{n/N}\log n)\) is the natural scale: the distinguished point
lies only \(\Theta(N/n)\) outside \([-1,1]\), where a degree-\(k\) recurrence
polynomial grows as \(\exp(\Theta(k\sqrt{N/n}))\). Thus obtaining polynomial
growth in \(n\) requires the displayed logarithmic factor. This does not rule
out improvements by a different method.

\section*{Acknowledgments}

The author thanks Zachary Chase for helpful comments, and thanks Kenneth Regan, Pepe Swer and Joseph Swernofsky for careful readings and suggestions that improved the presentation.

\section*{Declaration of AI Assistance}

The author discovered simplification potential when studying the analytical property of the periodic DFAs and the polynomial recurrences. GPT 5.5 was used
interactively to explore candidate simplifications, test proof variants, and
assist with drafting. The author rewrote the elementary proof. The author assumes responsibility for all content.

\end{document}